\documentclass[11pt,twoside]{article}
\usepackage{amsfonts}
\usepackage{amssymb}
\usepackage{amsmath}
\usepackage{graphicx}

\usepackage{amsthm}              %\usepackage{ntheorem} %CONFLICTS WITH amsthm
   \theoremstyle{plain}
   
   \newtheorem{theorem}{Theorem}[section]   %<-- If want Theorem 1.1, 2.1
   \newtheorem{claim}[theorem]{Claim}
   \newtheorem{lemma}[theorem]{Lemma}

\begin{document}
\title{Binary weights spanning trees and the $k$-red spanning tree problem in linear time}
\author{
Dorit S. Hochbaum \thanks{Department of Industrial Engineering and Operations Research,
University of California, Berkeley, email: {\tt dhochbaum@berkeley.edu}. Research supported in part by AI institute NSF award 2112533. }
}
\date{}
\maketitle

\begin{abstract}
We address here spanning tree problems on a graph with binary edge weights.  For a general weighted graph the minimum spanning tree is solved in super-linear running time, even when the edges of the graph are pre-sorted.  A related problem, of finding a spanning tree with a pre-specified sum of weights, is NP-hard.  In contrast, for a graph with binary weights associated with the edges, it is shown that the minimum spanning tree and finding a spanning tree with a given total sum, are solvable in linear time with simple algorithms.
\end{abstract}

\section{Introduction}
We address here binary weights spanning tree problems.   A given undirected graph has binary weights associated with the edges.  One of the problems is to find a minimum or maximum weight spanning tree.  For general weights, there are a number of minimum (or maximum) spanning tree algorithms, all of which have super-linear running times, even when the edges of the graph are pre-sorted.  On the other hand, when all the edge weights are equal, to $1$, then any search tree in the connected graph, a DFS or BFS trees, are found in linear time in the number of edges.  We show here, that the binary weights spanning tree, as well as the $k$-red spanning tree are solved in the same complexity as a spanning trees in unweighted graphs.

The study here is motivated by a simple problem I have recently included in a take-home final exam.
Given a connected graph $G=(V,E)$ and a partition of $E$ into $E_r \cup E_b$ where $E_r$ are red edges and $E_b$ are blue edges. Let $|V|=n$, $|E|=m$, $|E_r|=m_r$, and $|E_b|=m_b$. The $k$-red spanning tree problem has two versions: One is to determine whether there exists a spanning tree in $G$ that has exactly $k$ red edges; the second version is to construct a $k$-red spanning tree if one exists.

It appears that various writings, found online, on the $k$-red spanning tree problem confuse these two problems.  Further, the typically offered solution, described below, requires the generation of minimum spanning trees for existence, and applying exchange operations to generate the $k$-red spanning tree. Without sophisticated data structures, the exchange operations require quadratic time $O(n^2)$ and even with advanced data structures the known approaches require more than linear time.

It is shown here that the existence and construction of $k$-red spanning tree are both solved in linear time, as well as the generation of minimum (or maximum) spanning tree.

%It appears that the various writings on the k-red spanning tree problem confuse these two problems.  The interest in differentiating between the two versions, is that the existence problem is determined in linear time $O(m)$ while the constructive problem requires more than linear time with known data structures.

\section{An approach for finding a $k$-red tree and why it is wasteful}
The standard approach for verifying existence is to assign the red and blue edges the weights of $0$ and $1$, respectively.  Then find a minimum spanning tree  in $G$, $T_{\min}$, and a maximum spanning tree in $G$, $T_{\max}$.  A $k$-red spanning tree exists if and only if the $T_{\min}$ contains at least $k$ red edges, or its weight is less than or equal to $n-1-k$, and $T_{\max}$ contains at least $n-1-k$ blue edges, or its weight is at least $n-1-k$.

The proof relies on the use of the {\em exchange} operation.  An {\em exchange} is an operation on a spanning tree $T= (V,E_T)$, which adds an edge $[i,j] \notin E_T$, and removes an edge $[p,q]\in E_T$ from the unique path $P_{ij}$ that connects $i$ to $j$ in $T$.

\begin{claim}
A $k$-red spanning tree exists if and only if $|T_{\min}\cap E_r| \geq k$, and $|T_{\max}\cap E_b| \geq n-1-k$
\end{claim}
\begin{proof}
The conditions on $T_{\min}$ and $T_{\max}$ are obviously necessary: If a $k$-red spanning tree exists then its weight is $n-1-k$.  Therefore $T_{\min}$ will include at least $k$ red edges and its weight can be only less than $n-1-k$. Similarly, the weight of $T_{\max}$ is at least the weight of $n-1-k$.   To show that the conditions are sufficient we use the exchange operations on the symmetric difference between $T_{\min}$ and $T_{\max}$. If the symmetric difference is empty then $T_{\max}$ contains exactly $k$ red edges, and we are done.  Otherwise consider a red edge $e$ in $T_{\min}\setminus T_{\max}$ and exchange it for an edge $e'$ in $T_{\max}\setminus T_{\min}$.  Such an edge $e'$ must exist on the unique cycle created by adding $e$ to $T_{\max}$ since otherwise $T_{\min}$ would have included a cycle contradicting its acyclic property.   Therefore each exchange operation creates a tree $T$ with one more edge in common with $T_{\max}$ as compared with the previous one.  Therefore the number of red edges in the created tree $T$ increases by $0$ or $1$.  Since $T_{\max}$ contains strictly more than $k$ red edges, after a number of exchanges, that does not exceed $|T_{\max}\setminus T_{\min}|$ exchanges, the number of red edges in the created tree is exactly $k$.
\end{proof}

Notice that this proof is constructive even though the conditions are not.  That is, it uses the exchange operations to construct the $k$-red spanning tree in fewer than $n$ exchanges.  We also note that the best known running time to construct the trees $T_{\min}$ and $T_{\max}$ is $O(m \alpha (n,m))$ where $\alpha (n,m))$ is the very slowly growing inverse-Ackerman function, \cite{Chazelle2000}.  With these conditions the running time for verifying existence is $O(m \alpha (n,m))$.  The implied construction however, with $O(n)$ exchange operations, require, with a straightforward algorithm, $O(n^2)$ steps.

We show next that existence can be verified and the $k$-red tree constructed, more efficiently and simply, in $O(m)$ time.

%Also, what is the best complexity of constructing the trees $T_{\min}$ and $T_{\max}$?\\
%Just generating a spanning tree in a connected graph takes $O(m)$, e.g.\ with BFS.  But generating a spanning tree that is minimum or maximum, even with binary weights, appear to require more running time.
\section{Existence and construction of a $k$-red tree in linear time}
We now state different necessary and sufficient conditions that are easily implemented in linear time and do not require the construction of trees.  The key is that to prove existence, it suffices to find all the connected red components and then to find the connected blue components within the red components.  The red components are connected in the red edges induced graph $G_r=(V,E_r)$.  Finding those components is done in linear time $O(m_r)$ using BFS (Breadth First Search), and outputs also the BFS spanning trees of all the red components.  Let these, say $p$, red connected components induce a partition of $V$ into the sets of nodes $V_1, \ldots ,V_p$.  Note that some of these components can be singletons and therefore contain no edges. \\
{\bf Condition 1:} $\sum _{i=1}^p |V_i| \geq k +p$.  This is required as the sum of the number of red edges in the spanning trees in all these components has to be at least $k$.  Putting it differently, $q=\sum _{i=1}^p \{|V_i|-1\} \geq k$.\\
{\bf Condition 2:} This condition verifies that there are enough blue edges in those red components to replace the surplus number of red edges.  This number of blue edges has to be at least $q-k$. We next find the connected blue components of blue edges inside each of the graph's red components induced by the sets of nodes of each red component $G_i= (V_i, E_b)$, $i=1,\ldots ,p$.  Let the blue connected components in $G_i$ induce the disjoint sets of nodes $V_{i,1},\ldots , V_{i,p(i)}$ all of which are contained in $V_i$, for $i=1,\ldots ,p$. Let $m_i(b)= \sum _{j=1}^{p(i)} \{ |V_{i,j}|-1\}$ be the number of blue edges in a largest spanning forest of $G_i$.  Condition 2 is the inequality $\sum _{i=1}^p m_i(b) \geq q-k$.

\begin{lemma}
There is a $k$-red spanning tree in $G$ if and only if conditions 1 and 2 are satisfied; or equivalently:  $\sum _{i=1}^p m_i(b) \geq q-k \geq 0$.
\end{lemma}
\begin{proof}
Firstly we note that the first inequality is condition 2 and the second inequality is condition 1.  Condition 1 follows since $q-k \geq 0$ is necessary to have at least $k$ red edges in a spanning tree. Condition 2 is correct by construction -  applying the exchange algorithm on $q-k$ blue edges of the spanning forest in the red connected components. If not satisfied, then any spanning tree will have strictly more than $k$ red edges.

\end{proof}

Testing condition 1 runs in $O(m_r)$ steps, and testing condition 2 runs in  $O(m_b)$ steps for a total of $O(m)$ complexity.  Therefore the existence of a $k$-red spanning tree is verified in linear time, $O(m)$.  We next show how to construct a $k$-red tree, if one exists, in linear time.

Consider red components $V_1, \ldots ,V_p$.  Shrinking each component of the $p$ components into a node and finding the blue edges spanning tree that spans those $p$ components, $T(b)$, of size $p-1$, can be done in linear time $O(m_b)$.  Finding the $p(i)$ blue connected components in $G_i$ that induce the disjoint sets of nodes $V_{i,1},\ldots , V_{i,p(i)}$ for all values of $i=1,\ldots ,p$, require a total of $O(m_b)$ steps.  In each such blue connected component on $V_{i,j}$, $j=1,\ldots , p(i)$, we have a spanning tree of blue edges. Hence in each red connected component there is a blue forest.

If $\sum _{i=1}^p |V_i| = k +p$ then the red BFS tree in each component plus the tree  $T(b)$ that link the red components form a $k$-red spanning tree, and we are done.  Suppose then that $\sum _{i=1}^p |V_i| > k +p$ so that $\sum _{i=1}^p |V_i| -\ell = k +p$ for $\ell \geq 1$ integer.  $\ell$ is the number of edges that must be removed from the union of red trees spanning the red components in order to leave exactly $k$ red edges.  Consider the smallest value of $t$ such that $\sum _{i=1}^t m_i(b) \geq \ell$.  In each of the first $t-1$ components consider the blue spanning forest on $m_i(b)$ blue edges and extend it, using BFS on the red edges, to a spanning tree of the component $i$, $i=1,\ldots ,t-1$.  In the component $G_t$ take any $\ell-\sum _{i=1}^{t-1} m_i(b)$ blue edges of the blue forest and extend them with red edges, using BFS, to a spanning tree of $G_t$.  The resulting spanning trees of all the red components with the tree $T(b)$ is a spanning tree in $G$ with exactly $k$ red edges.  Therefore this construction is accomplished in linear time $O(m)$.

\noindent
{\bf Data structure}\\
The input graph $G$ is provided with adjacency lists, where for each node $i$ we have:\\
red adjacency list (i); red-visited$(i)$; blue adjacency list (i); blue-visited$(i)$\\
The red adjacency list (i) are the nodes adjacent to $i$ along red edges. Similarly, blue adjacency list (i) are the nodes adjacent to $i$ along blue edges.  The flags: red-visited$(i)$ is {\sf nil} if node $i$ has not been visited as of yet in the BFS of the red nodes, or, if has been visited it is the index of the index $r(i)$ of the root node of the component BFS that has reached node $i$.

In addition to this list, there is a list of the red components, generated when they are found, each identified by the root node.  The component is a list of nodes in the red component.  Similarly, there is a list of blue components.

\noindent
{\bf Node shrinking.}\\
Here, when the red components are identified, they are shrunk in order to find the blue edges spanning tree that span the rest of the graph.  Shrinking a component does not require any modification in the data structure.  In general a shrunk collection of nodes means their adjacency lists are merged.  Here it is not needed.  Simply, when searching for a blue neighbor of a shrunk component, the blue adjacency lists of the respective nodes in the components are scanned one at a time.  It is possible that the lists overlap, in which case the adjacent node has already been identified as visited.  In any case, the scan for an adjacent blue node require at most a single pass of $E_b$.

\section{Discussion and extensions}
The $k$-red spanning tree can be viewed as a problem of seeking a spanning tree with a given sum of weights.  The result here demonstrates that this problem is polynomial time solvable for a graph with binary weights.  Consider then a graph with general weights.  The decision of whether a spanning tree of total sum of weights equal to $M$ is equivalent to deciding whether there exists a subset of $n-1$ weights, out of a list of $m$ weights, that adds up to $M$.  This is at least as hard as subset sum, and therefore at least weakly NP-hard.  It is however conjectured here to be strongly NP-hard.  Even for the case of a graph with edge weights in $\{ 0,1,2 \}$, we believe that finding a spanning tree of a given sum $M <2(n-1)$, is NP-hard.  The binary case discussed here is therefore significantly easier than these other weighted spanning tree problems.

%DESCRIBE HERE THE EXCHANGE MORE CAREFULLY.\\
%With the linear time conditions, one can generate, with BFS, and in linear time, the spanning forest of the red components.  Similarly, the spanning forest of the $G_i$ components is generated in linear time.  The next step requires to replace some red edges in the spanning forest of the red components by an appropriate number of blue edges.  This calls for an exchange operation.
%What is known about the complexity of an exchange?  It is not known to be done in constant time.  A straightforward method would traverse the path $P_{ij}$ in the tree when blue edge $[i,j]$ is added and select any edge on the cycle created which is red.  Scanning the path can take $O(n)$ steps in the worst case, which could lead to $O(n)$ running time per exchange.

\end{document}